\documentclass[aps,superscriptaddress,twocolumn,twoside,floatfix,10pt]{revtex4-1}
\usepackage{graphicx,amsmath,amssymb,amsfonts,color,array,enumerate}
\usepackage[normalem]{ulem}
\usepackage{times}
\usepackage[T1]{fontenc}
\usepackage{graphicx,amsmath,amssymb,amsfonts,amsthm,color,array}
\usepackage[ocgcolorlinks,colorlinks=true,linkcolor=blue,citecolor=red,linktocpage=true]{hyperref}
\usepackage{comment}
\usepackage{verbatim}

\newcommand{\ket}[1]{|#1\rangle}
\newcommand{\bra}[1]{\langle#1|}
\newcommand{\braket}[2]{\langle#1|#2\rangle}

\newcommand{\proj}[1]{|#1\rangle\langle#1|}

\newcommand{\ie}{{\it{i.e.}~}}

\newcommand{\rA}{\mathrm{A}}
\newcommand{\rB}{\mathrm{B}}
\newcommand{\rE}{\mathrm{E}}

\newcommand{\obs}{\mathrm{obs}}
\newcommand{\id}{\mathbb{I}}
\newcommand{\rg}{\mathrm{guess}}

\newcommand{\LHS}{\mathrm{LHS}}

\DeclareMathOperator{\tr}{tr}

\newtheorem{myfact}{Fact}
\newtheorem{myresult}{Result}

\begin{document}

\title{Maximal randomness expansion from steering inequality violations using qudits}

\author{Paul Skrzypczyk}
\affiliation{H. H. Wills Physics Laboratory, University of Bristol, Tyndall Avenue, Bristol, BS8 1TL, United Kingdom}

\author{Daniel Cavalcanti}
\affiliation{ICFO-Institut de Ciencies Fotoniques, The Barcelona Institute of Science and Technology, 08860 Castelldefels (Barcelona), Spain}

\begin{abstract}
We consider the generation of randomness based upon the observed violation of an Einstein-Podolsky-Rosen (EPR) steering inequality, known as one-sided device-independent randomness expansion. We show that in the simplest scenario -- involving only two parties applying two measurements with $d$ outcomes each -- that there exist EPR steering inequalities whose maximal violation certifies the maximal amount of randomness, equal to $\log(d)$ bits. We further show that all pure partially entangled full-Schmidt-rank states in all dimensions can achieve maximal violation of these inequalities, and thus lead to maximal randomness expansion in the one-sided device-independent setting. More generally, the amount of randomness that can be certified is given by a semidefinite program, which we use to study the behaviour for non-maximal violations of the inequalities. 
\end{abstract}

 \maketitle

Randomness is an important resource which has wide-spread use, ranging from Monte Carlo simulations to cryptographic keys. Generating `good' randomness is a notoriously difficult task, where the notion of `good' depends precisely on the context -- whether it be the absence of subtle correlations which might lead to false conclusions in Monte Carlo simulations, or being perfectly uncorrelated from an adversary, and therefore private, in a cryptographic setting. In a classical, and therefore deterministic setting, there are two main approaches to generating randomness. The first is to use pseudo-random-number-generators, which are able to turn a small amount of initial randomness (a seed), in a much larger string of numbers which appear random, under reasonable assumptions about the computational power available to test their quality. The second is to use chaotic systems, whose long-time behaviour is essentially impossible to predict without perfect knowledge of the initial conditions. 

Quantum theory, as a fundamentally non-deterministic theory, provides an alternative route \cite{AM16,H-CG-E17,BAK+17}. As a simple example, which way a photon takes after passing through a balanced beam splitter is a fundamentally probabilistic event, and thus serves as a basic quantum random number generator. Quantum theory, as a fundamentally nonlocal theory \cite{Bel64}, also provides a second route to randomness, which is much stronger than the first. Nonlocality is necessarily accompanied by uncertainty, with the latter providing the mechanism by which nonlocal effects can be possible at all without leading to signalling, i.e. to observable influences at a distance, which are forbidden by relativity. Thus observing nonlocality ensures that randomness is present, and due to monogamy -- the fact that nonlocality cannot be simultaneously shared among multiple systems-- this randomness must also have an element of privacy. 

Bell nonlocality \cite{Bel64,BCP+14} is the most famous form of quantum nonlocality, and considers the correlations between distant measurement outcomes. From a modern perspective it is understood to constitute a device-independent (DI) form of nonlocality, since it does not rely on any characterisation or trust of the underlying quantum state or measuring devices used. Device-independent randomness expansion (DIRE) \cite{Col06,PAM+10} was the first application of nonlocality to randomness generation. A user, who has access to a pair of devices, uses an initial seed of randomness to choose the measurement settings in a Bell test. If a Bell inequality is violated, the measurement outcomes are guaranteed to contained randomness, with the exact amount being a function of the observed violation. Randomness is said to be expanded as an initial seed is converted into a larger amount of randomness \cite{footnote1}. The original scheme of Ref. \cite{PAM+10} was able to achieve quadratic randomness expansion, while later, more sophisticated schemes, have now been shown to achieve better expansion, including exponential and unbounded expansion \cite{CK11,VV12,CY13,MS16,CSW14,MS17,A-FRV16}. 

Einstein-Podolsky-Rosen (EPR) steering \cite{EPR35} is a second form of quantum nonlocality, that considers the correlations between measurement outcomes of one party, and the states prepared (or `steered') for a second party. From a modern perspective  it is understood to constitute a one-sided-device-independent (1SDI) form of nonlocality, since it only relies on the characterisation of one set of measuring devices \cite{WJD07}. Randomness expansion can also be considered in the 1SDI setting (1SDIRE), with the user generating randomness from the measurement outcomes of the uncharacterised device \cite{BCW+12,LTB+14,PCS+15}. Nonlocality in the form of EPR steering is more robust than Bell nonlocality to imperfections such as loss and noise, and this in turns translates into better robustness for 1SDIRE. Thus, if the use of the 1SDI paradigm can be justified over the DI setting, which will typically depend upon the specific details of any actual implementation, then one can expect to obtain advantages for randomness expansion. 

Here we demonstrate an advantage of using the 1SDI setting, by demonstrating that unbounded randomness expansion can be achieved using the simplest form of steering inequality (something that has not been observed in the DI scenario). We consider a linear steering test with only two choices of measurements that is naturally tailored to $d$-dimensional systems. We show that a maximal violation (which can be achieved using any pure full-rank entangled qudit) is only consistent with uniformly random outcomes, hence leading to $\log_2 d$ bits of private randomness certified in the 1SDI setting. The amount of randomness generated from the initial seed can thus be arbitrarily large as $d \to \infty$.

Consider a situation where a user has two devices, labeled $\rA$ and $\rB$. Device $\rA$ accepts an input $x$, labeling a choice of measurement, and produces an outcome $a$, labeling an outcome. The device accepts one out of $n$ inputs, $x \in \{0,\ldots, n-1\}$, and produces one out of $d$ outcomes, $a \in \{0,\ldots,d-1\}$. Apart from the observable  input-output behaviour of the device, no further characterisation will be assumed. Device $\rB$, on the other hand, will be assumed to be fully characterised. In particular, it will be assumed that the dimension of the system is known, and that known measurements can be performed. In particular, full tomography of the states of system $\rB$ could be performed, although in general this is not necessary. 

The information obtained in this scenario can be summarised by the conditional (unnormailsed) states prepared for system $\rB$, conditioned on the different measurement choices $x$ and outcomes $a$ of system $\rA$:
\begin{equation}
\sigma_{a|x}^\rB = \tr_\rA[(M_{a|x}\otimes \id)\rho^{\rA\rB}],
\end{equation} 
where $\rho^{\rA\rB}$ is the (unknown) state shared between the two devices and $M_{a|x}$ are the (unknown) measurement operators applied in $\rA$. The set $\{\sigma_{a|x}^B\}$ is usually refereed as to the assemblage \cite{Pus13,CS17}. Notice that one can recover both the conditional probability distributions, $p(a|x) = \tr[\sigma_{a|x}]$ and the normalised states $\rho^\rB_{a|x} = \sigma_{a|x}/p(a|x)$. 

EPR steering is observed when the assemblage cannot be explained by a classical mechanism, called a \emph{local-hidden-state} (LHS) model (see \cite{CS17} for a review), which is witnessed by the violation of EPR steering inequalities of the form

\begin{equation}
 \beta := \tr \sum_{a,x} F_{a|x}^B\sigma_{a|x}^B \leq \beta^\LHS.
\end{equation}
$\{F_{a|x}\}_{a,x}$ is a collection of Hermitian operators, that should be measured by device $\rB$ in the case of measurement $x$ and outcome $a$, and $\beta^\LHS := \max_{\sigma_{a|x}^\LHS} \tr\sum_{a,x} F_{a|x}\sigma_{a|x}^\LHS$ is the maximal value of $\beta$ that can be obtained by any classical assemblage. 

When an assemblage violates an EPR steering inequality (\ie $\beta>\beta^\LHS$) it is impossible that for all $a$ and $x$ that $p(a|x) \in \{0,1\}$. This means that the outcomes of system $A$ must contain some randomness. This randomness can be quantified in the following way \cite{PCS+15}. We assume the presence of an Eavesdropper (Eve) holding a measurement device $\rE$ which might share a tripartite state $\ket{\psi^{\rA\rB\rE}}\bra{\psi^{\rA\rB\rE}}$ with devices $\rA$ and $\rB$. Eve is assumed to know the shared state and the form of the measurements in $\rA$ and $\rB$. Eve's goal is to guess $\rA$'s outcomes when $x = x^*$, which happens successfully with probability
\begin{equation}\label{e:pguess}
	P_\rg(x^*) = \max_{\rho^{\rA\rB}_e, p(e), M_{a|x}} \sum_e p(e)p(a = e | x^*,e),
\end{equation}
where $\rho^{\rA\rB}_e$ is the state, labeled by $e$, that Eve distributes to the devices with probability $p(e)$, and without loss of generality Eve will guess $a = e$ as the outcome of the measurement $x = x^*$ for this particular state. Finally, 
\begin{align}
p(a=e |x^*,e) &= \tr[(M_{a=e|x=x^*} \otimes \id)\rho_e^{\rA\rB}],
\end{align}
is the probability that $a = e$ when the measurement $x = x^*$ is performed on the state $\rho^{\rA\rB}_e$. Crucially, in \eqref{e:pguess}, the maximisation takes place only over those strategies of Eve consistent with the observable data of the user, that is, given the observed violation $\beta^\obs$ of a steering inequality. This constraint is formally given by 
\begin{equation}
\tr\sum_{a,x} F_{a|x}\sum_e p(e) \tr_\rA[(M_{a|x}\otimes \openone) \rho_e^{\rA\rB}] = \beta^\obs.
\end{equation}
The guessing probability $P_\rg(x^*)$ quantifies the optimal probability with which Eve can guess the outcome of device $A$. Whenever the guessing probability is less than unity this implies that Eve cannot perfectly guess the outcome, and hence it is inherently probabilistic (even given Eve's side information). The randomness in the outcomes is quantified by the min-entropy
\begin{equation}\label{e:min entropy}
H_\mathrm{min}(x^*) = -\log P_\rg(x^*).
\end{equation}

Before proceeding to the main results, one final preliminary fact is needed, which concerns the uniqueness of probability distributions which can arise in the steering scenario. In particular, 

\begin{myfact} Consider two sets of linearly independent states in $\mathbb{C}^d$, $\{\ket{\phi_a}\}_a$ and $\{\ket{\lambda_i}\}_i$, both of which span the Hilbert space. Assume that the expansion coefficients $u_i^a$ and $v_a^i$ do not vanish for all $a$,$i$, where $\ket{\phi_a} = \sum_i u_i^a \ket{\lambda_i}$ and $\ket{\lambda_i} = \sum_a v_a^i \ket{\phi_a}$.  Then there exists unique vectors $\{q_a\}_a$ and $\{\lambda_i\}_i$ (up to normalisation) such that
\begin{equation}
\sum_a q_a \proj{\phi_a} = \sum_i \lambda_i \proj{\lambda_i}.\label{e:ensemble}
\end{equation} 
In particular, up to normalisation, these vectors are given by
\begin{align}\label{e: prob solutions}
q_a &\propto \frac{\braket{\psi_a}{\lambda_i}}{\braket{\phi_a}{\omega_i}},& \lambda_i &\propto \frac{\braket{\phi_a}{\omega_i}}{\braket{\psi_a}{\lambda_i}},
\end{align}
where $\{\ket{\psi_a}\}_a$ and $\{\ket{\omega_i}\}_i$ are the (unique) dual sets of vectors with respect to $\{\ket{\phi_a}\}_a$ and $\{\ket{\lambda_i}\}_i$, satifying $\braket{\psi_b}{\phi_a} = \delta_{ab}$, $\braket{\omega_j}{\lambda_i} = \delta_{ij}$, which always exist due to the linear independence of the original sets \cite{footnote3}. 
\end{myfact}
\begin{proof}
This claim is proved by left-muliplying \eqref{e:ensemble} by $\bra{\psi_b}$ and right-multiplying by $\ket{\omega_j}$. Note that the denominator of each expression is one of the (complex conjugate) expansion coefficients, $\braket{\phi_a}{\omega_i} = (u_i^a)^*$  and $\braket{\psi_a}{\lambda_i} = (v_a^i)^*$ and hence are non-vanishing by assumption.
\end{proof}

In the case where $q_a \geq 0$ $\forall a$ and $\lambda_i \geq 0$ $\forall i$, then \eqref{e:ensemble} provide two ensemble decompositions of the same density operator. The above fact says that given only the vectors in these two ensembles, the probabilities are uniquely specified. It also says that two sets of linearly independent vectors with non-vanishing expansion coefficients uniquely specify a density operator.

In what follows it will be shown that obtaining the maximal violation of a certain EPR steering inequality involving only two measurements leads to maximal randomness generation. In particular, we will prove the following result: 
\begin{myresult} Consider a steering scenario where device $A$ accepts 2 inputs and produces one of $d$ outcomes, preparing states for system $B$ in $\mathbb{C}^d$. Consider an EPR steering functional with elements
\begin{equation}\label{e:functional}
F_{a|x} = \proj{\phi_{a|x}}
\end{equation}
where $\{\ket{\phi_{a|x}}\}_a$ is a linearly independent set of $d$ states in $\mathbb{C}^d$, for both values of $x$, and such that the expansion coefficients of one set in terms of the other is non-vanishing. The maximal value the steering functional can take is $\beta = 2$, and when this value is observed, the amount of randomness certified for the input $x^*$ is
\begin{equation}
H_\textrm{min}(x^*) = -\log\max_a q_{a|x^*}
\end{equation}
where $\{q_{a|0}\}_a$ and $\{q_{a|1}\}_a$ are the unique probability distributions that satisfy $\sum_a q_{a|0} \proj{\phi_{a|0}} = \sum_a q_{a|1}\proj{\phi_{a|1}}$, as given by \eqref{e: prob solutions}.
\end{myresult}

\begin{proof}
First note that the maximal value of the steering functional $\beta = 2$ can only be achieved by an assemblage with elements $\sigma_{a|x} = q(a|x) \proj{\phi_{a|x}}$. To see this, note that the value of the functional for a general assemblage with elements $\sigma_{a|x} = p(a|x)\rho_{a|x}$ is $\beta = \sum_{a,x} p(a|x) \bra{\phi_{a|x}} \rho_{a|x} \ket{\phi_{a|x}}$. For every $a$ and $x$ such that $p(a|x) \neq 0$, it must be that $\bra{\phi_{a|x}} \rho_{a|x} \ket{\phi_{a|x}} = 1$, otherwise $\beta < 2$. The only choice of $\rho_{a|x}$ that satisfies this is $\rho_{a|x} = \proj{\phi_{a|x}}$. For $a$ and $x$ such that $p(a|x) = 0$, the choice of $\rho_{a|x}$ is arbitrary and can therefore be chosen to be $\proj{\phi_{a|x}}$ without loss of generality.

From Fact 1, it follows that the $q(a|x)$ are uniquely determined, as the only pair of probability distributions $\{q_{a|0}\}_a$ and $\{q_{a|1}\}_a$ which satisfy $\sum_a q_{a|0} \proj{\phi_{a|0}} = \sum_a q_{a|1} \proj{\phi_{a|1}}$ \cite{footnote4}. 

Turning our attention to randomness generation, when a violation $\beta = 2$ is observed, the above implies that the most general strategy of Eve is to prepare assemblages of the form
\begin{equation}\label{e:form 2}
\sigma_{a|x}^e = q(a,e|x)\proj{\phi_{a|x}},
\end{equation}
such that $\sum_e q(a,e|x) = q_{a|x}$. Indeed, the reduced assemblage of the devices of the user, $\sum_e \sigma_{a|x}^e$, by virtue of attaining a maximal violation, from the above must have elements of the form $q_{a|x}\proj{\phi_{a|x}}$. The only way for a sum of operators to be rank-1 is for each element to be proportional to the same rank-1 element, and hence the claim follows. The non-signalling constraint from \eqref{e:pguess} then takes the form
\begin{equation}\label{e:pguess nosig}
\sum_a q(a|e,0)q(e) \proj{\phi_{a|0}} = \sum_a q(a|e,1)q(e)\proj{\phi_{a|1}},
\end{equation}
where we have used no-signalling to write $q(a,e|x) = q(a|e,x)q(e|x) = q(a|e,x)q(e)$. However, \eqref{e:pguess nosig} has the same form as \eqref{e:ensemble} from Fact 1, and hence it must be the case that $q(a|e,x) = q_{a|x}$, due to the uniqueness of the distributions. Crucially this shows that $a$ is conditionally independent of $e$. Therefore, Eve's guessing probability in this case is a simple optimisation over the probability distribution $\{p(e)\}_e$,  given by
\begin{equation}
P_\rg(x^*) = \max_{\{q(e)\}_e} \sum_e p(e) q_{e|x^*} = \max_e q_{e|x^*}.
\end{equation}
Using the definition of the min-entropy \eqref{e:min entropy} the result follows. 
\end{proof}

This shows that Eve can do no better than guess the most probable outcome of device $A$, which is the same as could be achieved without the use of quantum theory.  Moreover, using the above, by considering a situation where $q_{a|x^*} = 1/d$, i.e. where the only assemblage consistent with the violation of the inequality has a uniformly random outcome for the measurement $x^*$, then $H_\mathrm{min}(x^*) = \log(d)$ bits of randomness are certified in a 1SDI scenario. In what follows we show that this can be achieved by making appropriate measurements on on all pure partially entangled Schmidt-rank-$d$ states. 

In particular: 
\begin{myresult}
Consider a pure partially entangled Schmidt-rank-$d$ state in $\mathbb{C}^d \otimes \mathbb{C}^d$, given by $\ket{\Psi} = \sum_i \sqrt{\lambda_i} \ket{i}\ket{i}$, where $\lambda_i > 0$ for all $i$, and $\sum_i \lambda_i = 1$ are the Schmidt coefficients. Consider two measurements with elements $M_{a|0} = \proj{a}$ and $M_{a|1} = F\proj{a}F^\dagger$, where $F$ is the $d$-dimensional discrete Fourier transform. Finally consider an EPR steering functional with elements $F_{a|0} = \proj{a}$ and $F_{a|1} = \proj{\chi_a}$, where $\ket{\chi_a} = \sum_i \sqrt{\lambda_i} \bra{a}F^\dagger \ket{i} \ket{i}/\sqrt{d}$. Then, 
\begin{enumerate}[(i)]
\item Using the measurements $\{M_{a|x}\}_a$ on the state $\ket{\Psi}$ leads to an assemblage which maximally violates the EPR steering functional with elements $F_{a|x}$, i.e. achieves $\beta = 2$.
\item The outcome probabilities are uniformly random for the second measurement of device A,  $p(a|1) = 1/d$ for all $a$.
\end{enumerate}
Together, the above two facts imply that maximal randomness can be certified using this 1SDI randomness certification scheme, $H_\textrm{min}(x = 1) = \log d$. 
\end{myresult}
\begin{proof}
Performing the measurements with elements $M_{a|x}$ on the state $\ket{\Psi}$ leads to an assemblage with elements $\sigma_{a|0} = \lambda_a \proj{a}$ and $\sigma_{a|1} = \frac{1}{d}\proj{\chi_a}$, where we used the fact that $|\bra{i}F\ket{a}|^2 = 1/d$ for all $a,i$ to evaluate $p(a|1) = \tr[F\proj{a}F^\dagger \otimes \mathbb{I})\proj{\Psi}] = \sum_i \lambda_i |\bra{i}F\ket{a}|^2 = 1/d$. This demonstrates the second claim. Direct calculation shows that this assemblage achieves the value $\beta = 2$ for the EPR steering functional given, proving the first claim. 

The set $\{\ket{a}\}_a$ forms an orthonormal basis for $\mathbb{C}^d$. The set $\{\ket{\chi_a}\}_a$ form a non-orthogonal basis for $\mathbb{C}^d$ with dual basis $\ket{\theta_a} = \sqrt{d}\sum_i1/\sqrt{\lambda_i} \bra{i}F^\dagger \ket{a} \ket{i}$. It follows that the expansion coefficients of $\ket{\chi_a}$ in terms of $\ket{a}$, and vice-versa, are non-vanishing. Result 1 can thus be applied, since all of the required conditions hold. In conjunction with the fact that $p(a|1) = 1/d$, this leads to
\begin{equation}
H_\textrm{min}(x=1) = -\log \max_a p(a|1) = \log d.
\end{equation} 
\end{proof}

There are two points worth noting. First, if we consider maximally entangled states, where $\lambda_i = 1/d$, then $p(a|0) = 1/d$ also. In this case, one can naturally consider obtaining randomness from both inputs. Moreover, in this case $F_{a|x} = (\proj{\phi_{a|x}})^\intercal$ form a measurement for each $x$, and hence only two measurements need to be performed at $\rB$, as opposed to $2d$ different measurements in the general case.

Second, if in the above we were to replace the Schmidt-rank-$d$ state with a Schmidt-rank-$k$ state, for $k < d$, then the  analysis can still be applied on the support of the reduced state of system B, which will be rank-$k$, and $H_\text{min}(x=1) = \log k$ bits of randomness will be certified, the maximal possible using projective measurements for such a state. 

Note that in the above analysis we have performed only the ideal analysis, assuming infinite statistics. To implement the above in practice, with only finite statistics, the protocol of \cite{PAM+10}, which is outlined in full detail in \cite{PM13} from the case of DIRE can be applied in this 1SDIRE setting. 

In total, the above demonstrates the power of the 1SDI paradigm for randomness certification. The steering functionals  presented consistute the simplest possible functionals, comprising only two choices of measurement for the uncharacterised/untrusted device. Nevertheless, they are powerful enough to generate the maximal amount of randomness possible when considering projective measurements on a partially entangled state. By allowing the local dimension $d$ to tend to infinity, this scheme can generate an unbounded amount of randomness from this simple scheme.

So far, we have only considered the case of a perfect violation of the steering inequalities, and shown that this leads to maximal randomness certification. Since a perfect violation can never be observed in practice, it is also important to analyse what happens for an arbitrary violation $\beta < 2$. This can be carried out efficiently numerically by solving the SDP presented in the Supplementary Information. This method easily allows one to consider dimensions up to $d = 32$ on a standard desktop computer \cite{footnote5}. As an illustration, in Fig.~\ref{f: exp} we plot the amount of certified randomness as a function of $\beta^\obs$ for $d = 2$ to $14$. 

\begin{figure}[t!]
\centering 
\includegraphics[width=\columnwidth]{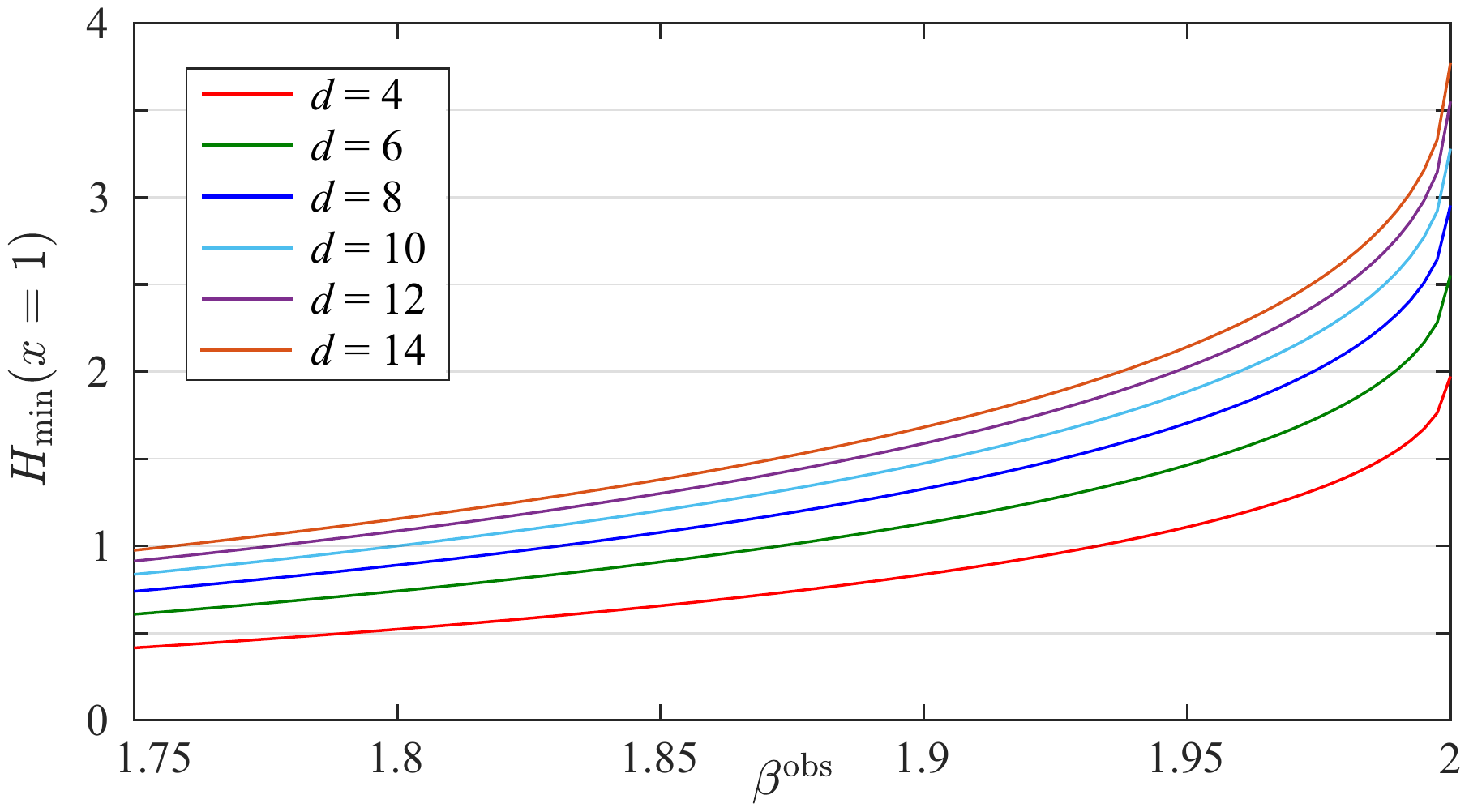}
\caption{Certified randomness of the second measurement, $H_{\min}(x=1)$, as a function of the observed violation of the steering inequality, $\beta^\obs$, for dimension $d = 4$ to $d = 14$. For any fixed violation, the amount of randomness that can be certified increases with dimension $d$. The code used to generate this figure is available at \cite{code}.} 
\label{f: exp}
\end{figure}

In \cite{WPD+18} the steering inequalities presented here were recently tested. In particular, \cite{WPD+18} presented an integrated silicon quantum photonic device with path encoded qudits up to $d = 16$. In the device each photon be prepared in a superposition over up to 16 spatial modes and entanglement can be generated between a pair of photons using coherent and controllable excitation of $d$ integrated identical photon pair sources. Arbitrary projective measurements can also be performed, using an integrated reconfigurable interferometric network. 

This device was used to prepare maximally entangled states in dimensions $d=2$ to $d = 16$, and a steering inequality of the form \eqref{e:functional} was tested, using as the two bases the path basis, and its Fourier transform. The authors demonstrated that higher dimensional systems require lower visibilities (lower inequality violation) to achieve the same amount of randomness.

To conclude, in this work we have considered the task of one-sided-device-independent randomness expansion. We have presented a general construction based on steering functionals in arbitrary dimension in the simplest scenario (consisting of only two inputs to the uncharacterised device). We have shown that a maximal violation of the associated steering inequality certifies that the outcomes of the uncharacterised device are completely unpredictable, even for a potential adversery, and hence maximal private randomness can be certified. We have shown that for every entangled state it is possible to construct a steering functional using our construction that is maximally violated, and hence that all entangled states in arbitrary dimension lead to maximal randomness certification using the simplest possible 1SDI scheme.

In the case of non-maximal violation of the steering inequality the amount of randomness that can be certified in a 1SDI manner can be computed using the technique of semidefinite programming. This provides a feasible method for dimensions up until $d \leq 32$. An important open problem is to obtain analytic lower bounds on the amount of randomness that can be obtained for near perfect violation which apply for arbitrary dimension $d$. This will provide a practical solution for arbitrary experimental situations. A route to achieve this would be to generalise Fact 1 to allow for some uncertainty in the sets of states, and to see how much freedom this allows for in the associated probabiliy distributions, which might be of independent interest. 

In a related direction, it would be interesting to understand what roll loss plays for 1SDI randomness expansion. The inequalities considered here are closely related to those put forward in \cite{SC15} for loss-tolerant EPR steering demonstrations. It would be interesting to extend the analysis here to these inequalities. 

Finally, the construction here is tailored to projective measurements, and hence in dimension $d$, up to $H_\text{min}(0) = \log d$ bits of randomness are certified. By using generalised positive-operator-valued (POVM) measurements it is in principle possible to certify up to $H_\text{min}(x^*) = 2\log d$ bits of randomness, by using measurements with $d^2$ outcomes. An interesting open question is whether the construction presented here can be generalised to this case also. 
 
\emph{Acknowledgements.--}
DC acknowledges the Ram\'on y Cajal fellowship, Spanish MINECO (QIBEQI FIS2016-80773-P and Severo Ochoa SEV-2015-0522), Fundaci\'o Cellex, and Generalitat de Catalunya (SGR875 and CERCA Program). PS acknowledges support from a Royal Society URF (UHQT).

\begin{appendix}

\section{SDP formulation of the guessing probability}
As shown in \cite{PCS+15}, it is possible to re-express the guessing probability presented in the main text as a SDP \cite{BV04,footnote2}. By defining $\tilde{\rho}_e^{\rA\rB} := p(e)\rho^{\rA\rB}_e$ the subnormalised state sent by Eve, and $\sigma_{a|x}^e := \tr_\rA[(M_{a|x}\otimes \openone)\tilde{\rho}^{\rA\rB}_e]$ the subnormalised assemblage, then the $P_\rg(x^*)$ is equivalent to
\begin{align} \label{e:pguess sdp}
P_\rg(x^*) = \max_{\{\sigma_{a|x}^e\}}&\quad\tr\sum_e\sigma_{a=e|x=x^*}^e \\
\text{s.t.}&\quad \tr\sum_{a,x}F_{a|x}\sum_e \sigma_{a|x}^e =\beta^\obs, \nonumber \\
&\quad \sum_a \sigma_{a|x}^e = \sum_a \sigma_{a|x^*}^e \quad \forall e,x, \nonumber \\
&\quad \tr \sum_{ae} \sigma_{a|x^*}^e = 1, \quad \sigma_{a|x}^e \geq 0 \quad \forall a,e,x. \nonumber
\end{align}
The first constraint enforces consistency of the average assemblage prepared by Eve with the observed steering inequality violation; the second enforces no-signalling, which arises from the fact that $\sum_a M_{a|x} = \openone$ for all $x$, satisfied by all valid measurements; the third enforces that $\sum_e p(e) = 1$; the last constraint simultaneously enforces that the $p(e) \geq 0$ and that the states prepared for system $B$, $\rho^e_{a|x}$ are positive semidefinite operators. In particular, it was shown in \cite{PCS+15} that given any set of assemblages $\{\sigma_{a|x}^e\}_e$ satisfying the above SDP, then one can always find a quantum strategy for Eve $\{p(e), \rho^{\rA\rB}_e, M_{a|x}\}$ which realises them, allowing Eve to guess $A$'s outcomes with the same guessing probability.
\end{appendix}

\begin{thebibliography}{100}

\bibitem{AM16} A. Ac\'{i}n and L. Masanes, \textit{Certified randomness in quantum physics}, Nature \textbf{540} 213 (2016).

\bibitem{H-CG-E17} M. Herrero-Collantes, J. C. Garcia-Escartin, \textit{Quantum random number generators}, Rev. Mod. Phys. \textbf{89}, 015004 (2017).

\bibitem{BAK+17} M. N. Bera, A. Ac\'in, M. Ku\'s, M. W. Mitchell, and M. Lewenstein, \textit{Randomness in quantum mechanics: philosophy, physics and technology}, Rep. Prog. Phys. \textbf{80}, 124001 (2017).

\bibitem{Bel64} J. S. Bell, \textit{On the Einstein Podolsky Rosen paradox}, Physics \textbf{1}, 195 (1964).

\bibitem{BCP+14} N. Brunner, D. Cavalcanti, S. Pironio, V. Scarani, S. Wehner, \textit{Bell nonlocality}, Rev. Mod. Phys. {\bf86}, 419 (2014).

\bibitem{Col06} R. Colbeck, PhD thesis, University of Cambridge (2006), arXiv:0911.3814 (2009).

\bibitem{PAM+10} S. Pironio et al., \textit{Random numbers certified by Bell's theorem}, Nature \textbf{464} (7291), 1021-1024 (2010)

\bibitem{footnote1} In fact, under certain natural assumptions \cite{Pir15}, the initial source of randomness can be public, while the randomness of the outcomes is private, hence one can even talk of randomness generation. Here, for simplicity of presentation, we will also talk about randomness expansion

\bibitem{Pir15} S. Pironio, \textit{Random `choices' and the locality loophole},
http://arxiv.org/abs/1510.00248.

\bibitem{CK11} R. Colbeck and A. Kent, \textit{Private randomness expansion with untrusted devices}, J. Phys. A
\textbf{44}, 095305 (2011)

\bibitem{VV12} U. Vazirani and T. Vidick, \textit{Certifiable Quantum Dice.} Phil. Trans. R. Soc. A \textbf{370},
3432 (2012).

\bibitem{CY13} M. Coudron, M. H. Yuen, \textit{Infinite Randomness Expansion and Amplification with a Constant
Number of Devices}, http://arxiv.org/abs/1310.6755.

\bibitem{MS16} C. A. Miller and Y. Shi, \textit{Robust protocols for securely expanding randomness and distributing
keys using untrusted quantum devices}, J. ACM, \textbf{63}, 33 (2016).

\bibitem{CSW14} K. M. Chung, Y. Shi and X. Wu, \textit{Physical Randomness Extractors: Generating Random
Numbers with Minimal Assumptions}, http://arxiv.org/abs/1402.4797.

\bibitem{MS17} C. A. Miller and Y. Shi, \textit{Universal security for randomness expansion from the spot-checking
protocol}, SIAM J. Computing \textbf{46}, 1304 (2017).

\bibitem{A-FRV16} R. Arnon-Friedman, R. Renner and T. Vidick, \textit{Simple and tight device-independent security
proofs}, http://arxiv.org/abs/1607.01797.

\bibitem{EPR35} A. Einstein, B. Podolsky, N. Rosen, \textit{Can Quantum-Mechanical Description of Physical Reality Be Considered Complete?} Phys. Rev. \textbf{47}, 777 (1935).

\bibitem{WJD07} H. M. Wiseman, S. J. Jones, and A. C. Doherty, \textit{Steering, Entanglement, Nonlocality, and the Einstein-Podolsky-Rosen Paradox
}, Phys. Rev. Lett. {\bf98} (2007).

\bibitem{CS17} D. Cavalcanti and P. Skrzypczyk, \textit{Quantum steering: a review with focus on semidefinite programming}, Rep. Prog. Phys. \textbf{80} 024001 (2017).

\bibitem{BCW+12} C. Branciard, E.G. Cavalcanti, S. P. Walborn, V. Scarani, and H. M. Wiseman,\textit{One-sided
Device-Independent Quantum Key Distribution: Security, feasibility, and the connection
with steering}, Phys. Rev. A \textbf{85}, 010301(R) (2012).

\bibitem{LTB+14} Y. Z. Law, L. P. Thinh, J. D. Bancal, V. Scarani, \textit{Quantum randomness extraction for various levels of characterization of the devices}, J. Phys. A: Math. Theor. {\bf 47}, 424028 (2014).

\bibitem{PCS+15} E. Passaro, D. Cavalcanti, P. Skrzypczyk, and A. Ac\'{i}n, \textit{Optimal randomness certification in the quantum steering and prepare-and-measure scenarios}, New J. Phys. \textbf{17}, 113010 (2015).

\bibitem{Pus13} M. F. Pusey, \textit{Negativity and steering: a stronger Peres conjecture}, Phys. Rev. A \textbf{88}, 032313 (2013).

\bibitem{footnote3}Note that the right hand side of each equation depends on both $i$ and $a$. However, after normalising, it is indeed the case that the correct dependence is obtained, i.e. the right-hand-side of the left-hand equation becomes independent of $i$, and vice versa.

\bibitem{footnote4} Note that Fact 1 did not guarantee that these are probability vectors. Here we will restrict to situations where this is the case. These are in fact the only situations that can arise in a steering experiment.

\bibitem{PM13} S. Pironio, S. Massar, \textit{Security of practical private randomness generation}, Phys. Rev. A. \textbf{87}, 012336 (2013).

\bibitem{footnote5} Using using the {\sc cvx} package \cite{cvx} for {\sc Matlab}, less than 10 minutes was required to obtain a solution.


\bibitem{cvx} M. Grant, S. Boyd, \textit{CVX: Matlab Software for Disciplined Convex Programming}, version 2.1 \texttt{http://cvxr.com/cvx} (2004)

\bibitem{code} Code available at https://git.io/vxkmL

\bibitem{WPD+18} J. Wang \textit{et. al.}, \textit{Multidimensional quantum entanglement with large-scale integrated optics}, Science 10.1126/science.aar7053 (2018); arXiv:1803.04449

\bibitem{SC15} P. Skrzypczyk, D. Cavalcanti, \textit{Loss-tolerant Einstein-Podolsky-Rosen steering for arbitrary-dimensional states: Joint measurability and unbounded violations under losses}, Phys. Rev. A. \textbf{92}, 022354 (2015).

\bibitem{BV04} S. Boyd and L. Vandenberghe, \textit{Convex optimization}, Cambridge University Press (2004).

\bibitem{footnote2} In \cite{PCS+15} randomness certified directly by the assemblage was considered, however it is a straightforward variant to consider instead a fixed steering inequality violation. 

\end{thebibliography}
\end{document}